%% file: main.tex
\documentclass[runningheads]{llncs}
\usepackage[T1]{fontenc}
\usepackage{graphicx}
\usepackage{booktabs}
\usepackage[misc]{ifsym}
\newcommand{\corr}{(\Letter)}

\usepackage{framed}
\usepackage{amsmath, amssymb}
\usepackage[numbers]{natbib}

\newcommand{\E}{\mathbb{E}}

\newcommand{\Xhat}{{\widehat{X}}}

\newcommand{\Zhat}{{\widehat{Z}}}
\newcommand{\indic}[1]{\mathbf{1}\left\{#1\right\}}

\renewcommand{\l}{\left}
\renewcommand{\r}{\right}

\newcommand{\Var}{\operatorname{Var}}
\newcommand{\Cov}{\operatorname{Cov}}

\tocauthor{Winston Chou}
\toctitle{Blending Proxy Metrics with a North Star}

\newif\ifincludeappendix
\includeappendixfalse 

\begin{document}

\title{Blending Proxy Metrics with a North Star}


\author{Winston Chou \corr}

\authorrunning{W. Chou}

\institute{Netflix, New York, NY \email{wchou@netflix.com}}

\maketitle

\begin{abstract}
    Proxy metrics are widely used to improve the precision and velocity of online experimentation (aka A/B testing).  Although proxies are often motivated by long-term outcomes that the experimenter does not observe, in many settings they are used alongside a \emph{contemporaneous} but statistically insensitive north star.  This can lead to a practical dilemma: when should experimenters trust the proxy metric, and when should they trust the north star?  In this paper, I propose an optimal blending approach that smoothly guides decision-making towards the north star as the power of the experiment increases and away from the north star as the quality of the proxy metric improves.  I study the implications of this decision-making framework for the design of experiments and of experimentation programs.  Equipped with better (worse) proxy metrics, experimenters should run smaller and more (larger and fewer) experiments.  I show how to leverage past experiments to estimate optimal blending weights and experiment sizes.  Lastly, I describe the real-world application of the methodology to an experimentation program at Netflix.
\end{abstract}


\input{sections/01-introduction}
\input{sections/02-related-literature}
\input{sections/03-methodology}
\input{sections/04-simulations}
\input{sections/05-empirical-application-at-netflix}
\input{sections/06-conclusion}

\bibliographystyle{splncs04}
\bibliography{bib}

\ifincludeappendix
    \newpage
    \input{Supplementary/appendix-content}
\fi

\end{document}

%% file: sections/01-introduction.tex
\section{Introduction}

Proxy metrics are widely used to increase the pace and precision of digital experimentation.  As the name suggests, the goal of a proxy metric is to approximate a \emph{north star} metric that directly measures business success, such as customer conversions or user retention, with inputs that are more sensitive to product interventions, such as user engagement \citep{deng2016data}.

Proxy metrics, and related concepts like statistical surrogates, are often motivated by long-term outcomes that the experimenter cannot observe outside of observational data or rare long-term experiments \citep{athey2019surrogate, prentice1989surrogate, tripuraneni2024choosing}.  However, in many applications, experimenters \emph{do} observe the north star, but find that effects on it are too small to detect in everyday experiments.  For example, a digital subscription service observes whether members are retained at the end of an experiment, but the vast majority of digital product interventions (e.g., changes to recommendation algorithms) have very small effects on retention that can only be detected in very large experiments \citep{larsen2024statistical}.  This motivates the use of a proxy north star, such as the estimated treatment effect on retention as a function of the treatment effect(s) observed on user engagement \citep{bibaut2024learning, tran2023inferring}.

Thorny questions arise in this setting, for example:
\begin{enumerate}
    \item Should experimenters still rely on the proxy outcome in the large, infrequent experiments that \emph{are} powered to detect treatment effects on the true north star?
    \item What should experimenters do if the proxy north star and the true north star disagree?
    \item How should the quality of the proxy affect the design of experiments and experimentation programs?
\end{enumerate}

In this paper, I derive a methodology of deciding experiments that smoothly ``blends'' a proxy metric (or set of proxy metrics) with a north star metric.  Specifically, letting $a$ denote some number between 0 and 1 and $A := (a, 1 - a)^\top$, I consider the blended metric
\begin{equation}
    \Zhat(A) = A^\top \widehat{X} = a \widehat{P} +  (1 - a) \widehat{Y},
\end{equation}
where $\widehat{P}$ and $\widehat{Y}$ are estimated treatment effects on the proxy and north star, respectively, and $a$ is chosen to maximize the ``returns'' to $Y$ in experiments \citep{azevedo2018b, sudijono2024optimizing}.  Intuitively, the optimal $a^*$ (to be defined formally in Section \ref{sec:methodology}) should depend on the power of those experiments: as the power of an experiment goes to 1, $a^*$ should go to zero, and decisions should be made based solely on the north star.  Intuitively, $a^*$ should also depend on the quality of the proxy metric: at any fixed sample size, it should be increasing in the structural covariance of true treatment effects on the proxy with true treatment effects on the north star.  Lastly, the quality of the proxy should have implications for the design of future experiments.  If the proxy is excellent and large experiments are costly (for example, because they reduce traffic to other experiments), then experimenters should run smaller and more experiments compared to when the proxy is poor.

My analysis confirms these intuitions and motivates an estimator of $a^*$ that leverages the structural covariance between treatment effects on the proxy and the north star learned from past experiments \citep{bibaut2024learning, cunningham2022interpreting, tripuraneni2024choosing}.  Using simulations, I show how blending can significantly increase the returns to experimentation relative to making decisions exclusively based on the proxy metric ($a = 1)$ or the north star ($a = 0$).  Lastly, I describe the real-world application of the methodology to construct an adaptively blended decision metric for A/B tests at Netflix.

%% file: sections/02-related-literature.tex
\section{Related Literature}

This paper contributes first and foremost to the literature on proxy metrics in A/B testing, especially to research on how to learn good proxy metrics from meta-analysis of past experiments \citep{bibaut2024learning, buyse2000validation, chou2025evaluating, deng2024metric, deng2016data, tripuraneni2024choosing, zhang2023evaluating}.  Distinctively, I focus on settings in which experimenters simultaneously observe \emph{both} an insensitive north star metric \emph{and} more sensitive proxies for that metric.  A natural question in these settings is whether and how experimenters should incorporate both types of metrics into decision-making.

I show that the optimal blend of metrics in decision-making depends on the size of the experiment, with the optimal weight concentrating exclusively on the north star metric in an infinitely large experiment.  This echoes the insight of \cite{tripuraneni2024choosing} that the optimal decision rule for an experiment depends on its power.  Building on the return-aware framework of \cite{azevedo2018b, azevedo2023b, sudijono2024optimizing}, I draw out the implications of this insight when sample sizes are endogenous.  Specifically, experimenters should run smaller and more experiments if the proxy is excellent, and larger and fewer experiments if the proxy is poor.

%% file: sections/03-methodology.tex
\section{Methodology}
\label{sec:methodology}

\input{sections/03-methodology-subsections/03-01-framework-and-notation}
\input{sections/03-methodology-subsections/03-02-optimal-blending-weights}
\input{sections/03-methodology-subsections/03-03-estimation-from-past-experiments}
\input{sections/03-methodology-subsections/03-04-implications-for-statistical-power}
\input{sections/03-methodology-subsections/03-05-implications-for-experimentation-programs}

%% file: sections/03-methodology-subsections/03-01-framework-and-notation.tex
\subsection{Framework and Notation}

The setting of this paper is as follows.  Experiments are drawn independently from a distribution and consist of true treatment effects $X = (P, Y)$, where $P \in \mathbb{R}^m$ is a vector, possibly one-dimensional, of treatment effects on proxy metrics and $Y$ is the treatment effect on the north star metric.  Denote the variance-covariance matrix of $X$ across experiments by $\Sigma_X$.

The experimenter observes estimates of the treatment effects $\widehat{X} = (\widehat{P}, \widehat{Y})$ on the proxy metric(s) and north star metric.  For simplicity, I assume that each experiment consists of a treatment group and a control group, each of size $n$ in every experiment, and that the experimenter estimates treatment effects using the difference in sample means.  Thus, conditional on the true effects, the estimated treatment effects $\widehat{X}$ are Gaussian with distribution
\begin{equation}
    \widehat{X} | X \sim \mathcal{N}(X, 2\Omega/n),
\end{equation}
where $\Omega$ is the sampling variance of the unit-level random variables underlying $Y$ and $P$ \citep{cunningham2022interpreting}.  I assume that $\Sigma_X$ and $\Omega$ are both positive-definite.  Therefore, the unconditional (total) variance-covariance matrix of $\widehat{X}$, $\Sigma_{\widehat{X}} := \Sigma_X + 2\Omega/n$, is also positive-definite.

I further assume that the true treatment effects $X$ are Gaussian with zero mean---a workhorse model in the literature that facilitates closed-form analysis \citep[][]{azevedo2023b, cunningham2022interpreting, tripuraneni2024choosing}.  In Section~\ref{sec:simulations}, I show using simulations that the key takeaways hold for heavier-tailed treatment effect distributions \citep[cf.][]{azevedo2018b}.

%% file: sections/03-methodology-subsections/03-02-optimal-blending-weights.tex
\subsection{Choosing Optimal Blending Weights}
\label{sec:optimal-weights}

I propose a \emph{return-aware} approach to choosing the optimal blending weights \citep{azevedo2018b, sudijono2024optimizing}.  Specifically, I derive the optimal choice of $A$ under a decision rule that ``launches'' the treatment so long as the blended statistic $\Zhat(A) := A^\top\Xhat$ is significantly greater than zero at some chosen significance level $\alpha$.  The returns to $Y$ of such a decision rule are given by
\begin{equation}
    R(A) := \mathbb{E}\l[Y\indic{A^\top\Xhat > z\sqrt{A^\top (2\Omega/n) A}}\r],
    \label{eq:return-operational}
\end{equation}
where $z$ is the critical value corresponding to $\alpha$.

$A^*$ is defined as the choice of $A$ that maximizes \eqref{eq:return-operational}
\begin{equation}
A^* := \arg \max_{A \in \mathcal{A}} R(A),
\end{equation}
where $\mathcal{A} = \{A \in \mathbb{R}^{m+1} : A \ge 0, \sum_{i=1}^{m+1} A_i = 1\}$.  I will assume there is an interior solution to this optimization problem (i.e., $A^*$ contains strictly positive weights).  This assumption holds by construction if, for example, one evaluates a set of proxies thought to have a positive relationship with the north star and then, for interpretability, keeps only those proxies having strictly positive weights at the optimum.  Note that this does not rule out proxies that are thought to have a negative relationship with the north star (e.g., customer service complaints), since these can be rescaled by $-1$; it only rules out proxies with zero or ``wrongly'' signed weights at the optimum.

The optimal $A^*$ can be found by solving the following constrained optimization problem:
\begin{equation}
\max_{A \in \mathcal{A}} R(A) = \max_{A \in \mathcal{A}} \frac{A^\top\Gamma}{\sqrt{A^\top\Sigma_{\widehat X}A}}\,\phi\!\big(z\sqrt{q_n(A)}\big),
\label{eq:optimization-problem}
\end{equation}
where $\Gamma := \Cov(X, Y) = \Cov((P, Y), Y)$, $\phi$ is the standard normal pdf, and $q_n(A) := \frac{A^\top (2\Omega/n) A}{A^\top \Sigma_\Xhat A}$ is the ratio of the within-experiment sampling variance of $\Zhat(A)$ to its total variance across experiments.  Although \eqref{eq:optimization-problem} does not have a closed-form solution in general, Proposition~\ref{prop:approximation} provides a closed-form approximation to $A^*$ that is very accurate for sufficiently large experiments.

\begin{proposition}
    The \emph{approximately optimal blending weighting vector} $\tilde{A}$ is defined as the vector on the simplex that satisfies
    \begin{equation}
        \tilde{A} \propto (\Sigma_X + (1 + z^2) 2\Omega/n)^{-1} \Gamma.
    \end{equation}

    Assume joint normality of $X = (P, Y)$ and that $A^*$ is an interior solution to \eqref{eq:optimization-problem}, i.e., $A_i^* > 0$ for all $i$.  Then, for $n \ge \underline{n}$, where $\underline{n}$ is given by \eqref{eq:bound-n}, we have that:
    \begin{enumerate}
        \item The approximation error of $\tilde{A}$ is $O(n^{-2})$.
        \item The regret of $\tilde{A}$ is $O(n^{-4})$.
    \end{enumerate}
    \label{prop:approximation}
\end{proposition}

\begin{proof}
    For $n \ge \underline{n}$, the optimal blending weights $A^*$ solve the following fixed-point equation:
    \begin{equation}
    A^* \propto (\Sigma_X + \left(1 + \frac{z^2}{1 - z^2 q_n(A^*)}\right) 2\Omega/n)^{-1} \Gamma.
    \label{eq:fixed-point-weights}
    \end{equation}
    $q_n(A^*)$ is $O(n^{-1})$, so $z^2 / (1 - z^2 q_n(A^*)) = z^2 + O(n^{-1})$.  $2\Omega/n$ is also $O(n^{-1})$, so $(z^2 - z^2 + O(n^{-1})) 2 \Omega/n$ is $O(n^{-2})$.  Therefore, the direction of the optimal weights $(\Sigma_X + (1 + z^2) 2\Omega/n + O(n^{-2}))^{-1} \Gamma$ matches the direction $\tilde{A}$ through first order.  This establishes the first item.

    By a Taylor expansion of $R(A)$ around $A^*$, we have that
    \begin{equation}
    R(\tilde{A}) = R(A^*) + \nabla R(A^*)(\tilde{A} - A^*) + \frac{1}{2} (\tilde{A} - A^*)^\top H(A^*) (\tilde{A} - A^*) + o(||\tilde{A} - A^*||^2),
    \end{equation}
    where $\nabla R(A^*)$ and $H(A^*)$ are the gradient and Hessian of $R(A)$ at $A^*$, respectively.  Interiority of $A^*$ implies that $\nabla R(A^*)^\top (\tilde{A} - A^*) = 0$.  Therefore, the regret of $\tilde{A}$ is:
    \begin{equation}
        R(A^*) - R(\tilde{A}) = \underbrace{-\frac{1}{2} (\tilde{A} - A^*)^\top H(A^*) (\tilde{A} - A^*)}_{O(n^{-4})}.
    \end{equation}
    This establishes the second item. \hfill \qed
\end{proof}

In other words, Proposition~\ref{prop:approximation} states that $\tilde{A}$ is an excellent closed-form approximation to $A^*$ for sufficiently large experiments, where ``sufficiently large'' is defined by \eqref{eq:bound-n} below.  In turn, $\tilde{A}$ can be interpreted as the \emph{penalized} linear regression of $Y$ on $X = (P, Y)$ with the penalty given by a scale factor of the measurement error variance $(1 + z^2) 2\Omega/n$.  This penalty regularizes the weights towards metrics with smaller variance within experiments (i.e., $\Omega$) relative to their variance across experiments (i.e., $\Sigma_X$).  In an infinitely large experiment ($n \to \infty$), the penalty disappears, and the optimal weights concentrate exclusively on the north star metric.  However, at any finite value of $n$, the optimal weight on any given metric will, all else equal, be increasing in its respective component of $\Gamma$.

Note that the strength of the penalty is increasing in the chosen critical value $z$.  Intuitively, a large $z$ means that only the most extreme test statistics will result in launches, which drives up the optimal weights on proxy metrics that are statistically precise.  Because statistical precision tends to trade off with alignment with the north star, this can come at the expense of more structurally aligned proxies \citep{zito2025pareto}.

The accuracy of the approximation depends on the size of each treatment arm $n$.  The relevant condition for $n$ is that $q_n(A^*) \le 1 / z^2$, or equivalently that
\begin{equation}
    n \ge \underline{n}(A^*) := (z^2-1)\frac{2\,A^{*\top}\Omega A^{*}}{A^{*\top}\Sigma_X A^{*}}.
    \label{eq:bound-n}
\end{equation}
In the Appendix, I show that substituting $\tilde{A}$ for $A^*$ in this expression yields a weakly conservative bound, $\underline{n}(\tilde{A}) \ge \underline{n}(A^*)$.  Therefore, if $n$ is greater than $\underline{n}(\tilde{A})$, one can safely use the approximation.  If $n$ is smaller than $\underline{n}(\tilde{A})$, the optimal weights should be found by solving the constrained optimization problem \eqref{eq:optimization-problem} directly.

%% file: sections/03-methodology-subsections/03-03-estimation-from-past-experiments.tex
\subsection{Estimation From Past Experiments}
\label{sec:estimation}

Given a set of past experiments, $\tilde{A}$ can be estimated as follows:
\begin{enumerate}
    \item Estimate the within-experiment sampling covariance $2\Omega/n$, for example, by drawing a random sample from each experiment and estimating the sampling covariance from the concatenated samples.
    \item Estimate $\Sigma_X$ using nonparametric estimators \citep{bibaut2024learning}, empirical Bayes estimators \citep{azevedo2019empirical}, or a Bayesian hierarchical model \citep{tripuraneni2024choosing}.
    \item Estimate $\tilde{A}$ using the closed-form plug-in estimator
    \begin{equation}
        \dot{A} = (\widehat{\Sigma_X} + \left(1 + z^2\right) 2 \widehat{\Omega}/n)^{-1} \widehat{\Gamma},
    \end{equation}
    where $\widehat{\Gamma} = \widehat{\Cov}(X, Y) = \widehat{\Cov}((P, Y), Y)$ is the last column of $\widehat{\Sigma_X}$.
    \item If any component of $\dot{A}$ is negative, drop that component and repeat steps 1--3.
    \item Scale $\dot{A}$ by the sum of its components to obtain $\widehat{\tilde{A}}$.
    \item Check that $n \ge \underline{n}(\widehat{\tilde{A}})$ and, if not, solve the constrained optimization problem \eqref{eq:optimization-problem} numerically for the exact weights.
\end{enumerate}

I illustrate this estimator in Sections~\ref{sec:simulations} and~\ref{sec:emp}.

%% file: sections/03-methodology-subsections/03-04-implications-for-statistical-power.tex
\subsection{Implications for Statistical Power}
\label{sec:power}

Thus far, our focus has been on the impact of blending on the returns to experimentation.  In this section, I study how blending can affect another important aspect of experimentation programs: the statistical power of experiments.  Power is an important practical consideration, since experiments that are chronically underpowered can cause firms to miss valuable innovations, overstate the value of ``winning'' treatments \citep{gelman2014beyond, lee2018winner}, and waste resources on false positives \citep{kohavi2024false}.  Below, I derive the minimum sample size needed to detect a given treatment effect on the north star metric when using the blended metric $\Zhat(A)$ compared to the north star or proxy exclusively.

As a benchmark, consider the scenario in which the experimenter only decides experiments based on the estimated treatment effect on the north star metric ($a = 0$).  Then, by a standard power analysis, the minimum sample size required to detect a given treatment effect $\tau_Y$ with $80\%$ power is
\begin{equation}
n_0 := \frac{(z + 0.84)^2 2\,\omega_Y}{\tau_Y^2},
\end{equation}
where, as before, $z$ denotes the critical value for the test and $\omega_Y$ is the variance of the unit-level variable underlying $Y$.

Alternatively, consider the minimum sample size required to be ``powered'' for the same effect on the north star metric when using the blended metric $\Zhat(A)$.  Note that the use of ``power'' here is somewhat informal as $\Zhat(A)$ is not an estimate of $Y$.  Rather, by power I mean that, at this sample size, $\Zhat$ is significantly greater than zero with $80\%$ probability when $Y = \tau_Y$.  Denote this (per-arm) sample size by $n_A$.  Fixing $Y$ at $\tau_Y$, $\Zhat(A)$ is Gaussian with mean
\begin{equation}
\mu_A(\tau_Y) := \E[\Zhat(A) | Y = \tau_Y] = \frac{1}{\sigma_Y^2} A^\top\Gamma \tau_Y,
\end{equation}
and residual variance
\begin{equation}
\sigma_A(\tau_Y)^2 := \Var(\Zhat(A) \mid Y=\tau_Y) = A^\top \Sigma_\Xhat A - \frac{(A^\top \Gamma)^2}{\sigma_Y^2}.
\end{equation}

Launching whenever $\Zhat(A) > z\sqrt{A^\top(2\Omega/n)A}$ and imposing a target power of $80\%$ requires
\begin{equation}
\mu_A(\tau_Y) = z\sqrt{A^\top(2\Omega/n)A} + 0.84\,\sigma_A(\tau_Y).
\end{equation}
It will be convenient to define
\begin{equation}
\omega_A := A^\top \Omega A \quad \text{and} \quad \theta := \frac{0.84 \sigma_A(\tau_Y)}{\mu_A(\tau_Y)}.
\end{equation}
In words, $\theta$ is the ratio of the residual standard deviation of $\Zhat(A)$ given $Y = \tau_Y$, multiplied by the power threshold ($0.84$), to its mean.  Increasing the target detection probability or residual standard deviation increases $\theta$, while a larger conditional mean reduces $\theta$. 

Then, in order to be powered for the same effect size on the north star metric when using the blended metric $\Zhat(A)$, it must be that
\begin{equation}
\mu_A(\tau_Y) \ge z \sqrt{\frac{2\,\omega_A}{n_A}} + 0.84\,\sigma_A(\tau_Y),
\end{equation}
or equivalently,
\begin{equation}
1 \ge \frac{z \sqrt{2\omega_A / n_A}}{\mu_A(\tau_Y)} + \theta.
\end{equation}
Solving for the minimum sample size $n_A$, we have that:
\begin{equation}
n_A \ge \frac{(z + 0.84)^2 2\,\omega_A}{\mu_A(\tau_Y)^2} \cdot I(\theta),
\label{eq:n-a}
\end{equation}
where
\begin{equation}
I(\theta) = \frac{z^2}{(z + 0.84)^2 (1 - \theta)^2}.
\end{equation}

The expression for $n_A$ is (intentionally) analogous to the expression for $n_0$.  Both expressions are inflated by the sampling variance of the underlying variables ($\omega_Y$ and $\omega_A$, respectively) and deflated by the expected effect size ($\tau_Y$ and $\mu_A(\tau_Y)$, respectively).  However, the central difference is that $n_A$ is inflated by an additional factor $I(\theta)$ that accounts for imperfect alignment between the proxy metric(s) and the north star.

Note that, as $\theta \to 1$, $I(\theta) \to \infty$.  Thus, $\mu_A(\tau_Y) > 0.84 \sigma_A(\tau_Y)$ is a necessary condition for the rule to attain $80\%$ power for some finite sample size.  This condition requires that both the conditional mean of $\Zhat(A)$ is sufficiently large and that its residual variance is sufficiently small given $Y = \tau_Y$.

Because both sides of $\eqref{eq:n-a}$ depend on $n_A$, it does not directly yield an explicit solution for $n_A$ and is mainly presented for interpretation.  In practice, one can solve for $n_A$ by evaluating a grid of values and choosing the smallest value satisfying the inequality~\eqref{eq:n-a}.

%% file: sections/03-methodology-subsections/03-05-implications-for-experimentation-programs.tex
\subsection{Implications for Experimentation Programs}

Above, the optimal blending weights were shown to depend on experiment sizes: as the experiment size grows, experimenters should shift the balance of decision-making away from the proxy metric towards the true north star.  Yet, as \eqref{eq:optimization-problem} shows, increasing $n$ also increases the expected returns to any individual experiment.\footnote{$\sqrt{A^\top \Sigma_\Xhat A}$ is decreasing in $n$, while $\phi(\cdot)$ attains its maximum at 0, which is the limit of $z \sqrt{q_n(A)}$ as $n \to \infty$.}  With just a single experiment and unlimited allocations, experimenters should allocate as many units as possible to their experiment and shift decision-making towards the north star.

However, in practice, experimenters are responsible for a portfolio of experiments, to which they have a finite number of units to allocate \citep{azevedo2018b,sudijono2024optimizing, tang2010overlapping}.  Therefore, increasing allocations to any one experiment has an opportunity cost, as it reduces the number of allocations available to other experiments.

To formalize this tradeoff, let $R(A, n) / n$ denote the average return of a single experiment of size $n$ (per arm) under any choice of $A$:
\begin{equation}
    R(A, n) / n =  \frac{\Gamma^\top A}{\sqrt{A^\top \Sigma_\Xhat A}} \phi\left(z \sqrt{q_n(A)}\right) / n.
\end{equation}

In the Appendix, we show that $R(A, n) / n$ is decreasing in $n$ for any choice of $A$ assuming typical values of $z$ (i.e., $1 \le z \le 1 + \sqrt{2} \approx 2.414$).\footnote{The Appendix for this paper can be found at \url{https://github.com/winston-chou/blending}.}  Note that this holds for any $A$, including the optimal $A$ given $n$.  Therefore, letting $\delta > 0$,
\begin{equation}
    R(A^*(n + \delta), n + \delta) / (n + \delta) \le R(A^*(n + \delta), n) / n \le R(A^*(n), n) / n,
\end{equation}
where $A^*(n) = \arg \max_A R(A, n)$.  Thus, letting $R^*(n) := R(A^*(n), n)$, the average optimal return per unit of allocation $R^*(n) / n$ is also decreasing in $n$.

Next, consider $n_A$ incremental treatment-control pairs (the minimum required to run a well-powered experiment).  If the experimenter allocates this block to a new test, the incremental expected return is $R^*(n_A)$.  However, if instead she spreads the same $n_A$ pairs evenly across $J$ existing powered tests of size $n_A$, each existing test receives $\delta = n_A/J$ additional pairs, yielding an incremental improvement of
\[
    J \left[ R^*(n_A+\delta)-R^*(n_A)\right].
\]
Thus, allocating the pairs to a new test is weakly preferred whenever
\begin{equation}
    R^*(n_A) \ge J \left[ R^*(n_A+\delta) - R^*(n_A) \right],
\end{equation}
which is directly implied by $\frac{R^*(n)}{n}$ decreasing in $n$.  Therefore, the experimenter always prefers to allocate additional units to a new test over distributing them over existing tests so long as the launch threshold is not too stringent (i.e., $z \le 2.414$, corresponding to a one-sided $\alpha$ of about 0.008).  If the launch threshold is more stringent than this, then average returns are no longer guaranteed to be decreasing in $n$, and the experimenter should consider distributing the additional units across existing tests to increase power rather than allocate to a new test.

%% file: sections/04-simulations.tex
\section{Simulations}

\label{sec:simulations}

To show how blending improves decision-making, here I conduct a simulation study that examines the performance of north star-only, proxy-only, and blended decision rules.  Replication code for all figures in this section can be found at \url{https://github.com/winston-chou/blending}.

I set the simulation parameters as follows:

\begin{framed}
    \paragraph{Simulation parameters.}
    \begin{itemize}
        \item Baseline true effect covariance: $\Sigma_X = \begin{bmatrix}1.0 & 0.08 \\ 0.08 & 0.01\end{bmatrix}$ (i.e., $\sigma^2_P = 1.0$, $\sigma^2_Y = 0.01$, $\sigma_{PY} = 0.08$, $\rho = 0.8$).  Simulations~2--4 fix $\sigma^2_Y = 0.01$ and $\sigma^2_P = 1.0$ but sweep $\rho = \mathrm{Corr}(P,Y) \in \{0, 0.2, 0.4, 0.6, 0.8, 1.0\}$.
        \item Sampling covariance: $\Omega = \begin{bmatrix}25 & 8 \\ 8 & 16\end{bmatrix}$ (i.e., $\omega_P = 25$, $\omega_Y = 16$).
        \item Critical value: $z \approx 1.645$ (i.e., one-sided $\alpha = 0.05$).
        \item Number of experiments per replicate for computing cumulative returns: $J = 100$.
        \item Per-arm sample size grid: $n \in \{1000, 2000, \ldots, 20000\}$.
        \item Monte Carlo replicates: $R = 1{,}000$.
    \end{itemize}
\end{framed}

These parameter values were chosen to reflect a realistic setting in which the north star metric is relatively insensitive to treatments, whereas the proxy metric is more sensitive in the sense of having a larger ratio of treatment effect variance to measurement error variance.  The closed-form approximation to the optimal blending weights from Proposition~\ref{prop:approximation} is used except where the bound $\underline{n}(\tilde{A})$ of Section~\ref{sec:optimal-weights} is violated---at low $\rho$ and small $n$---in which case the weights are obtained by numerically optimizing \eqref{eq:optimization-problem}.

First, I assess the cumulative returns over a set of $J$ experiments under each decision rule.  I plot these returns as a function of the per-arm sample size $n$ in the first panel of Figure \ref{fig:sim1-cumulative-returns}.  As the panel shows, the blended decision rule consistently yields higher cumulative returns than do the north star-only and proxy-only rules across all sample sizes.

\begin{figure}
    \centering
    \includegraphics[width=0.49\linewidth]{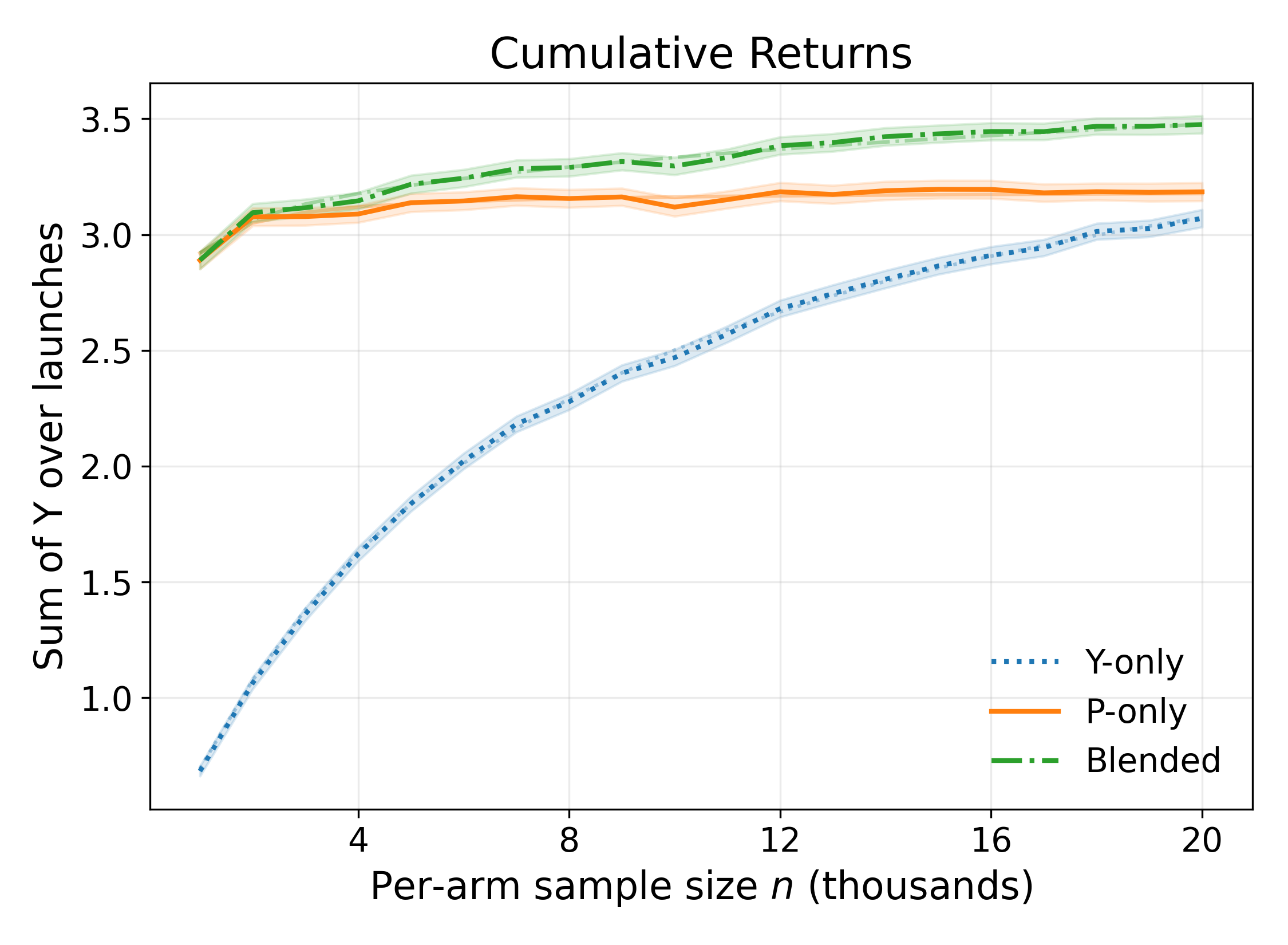}
    \includegraphics[width=0.49\linewidth]{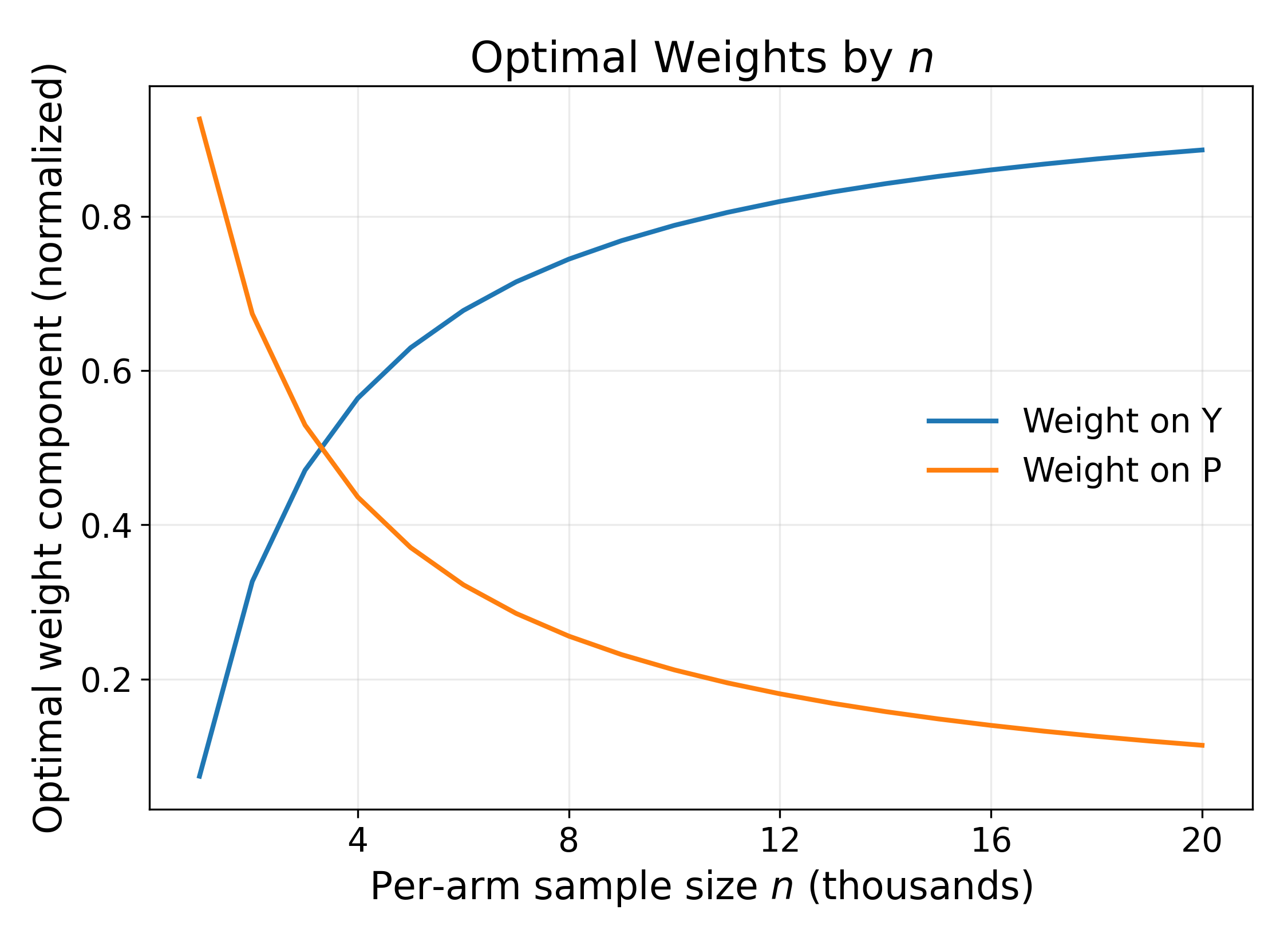}
    \caption{Cumulative Returns (left) and Optimal Blending Weights (right) by $n$.}
    \label{fig:sim1-cumulative-returns}
\end{figure}

The intuition for this is as follows.  At small sample sizes, the proxy metric dominates the north star due to its much greater signal-to-noise ratio.  However, as the sample size increases, treatment effects on the north star---which is better aligned with the objective---become measured with greater precision.  As a result, for very large $n$, it is better to shift decision-making towards the north star.  By construction, the blended metric uses the optimal weights at any given $n$, and so smoothly concentrates the weights on the north star as $n$ grows.  This behavior is shown in the second panel of Figure \ref{fig:sim1-cumulative-returns}.

Next, I evaluate the statistical power and false positive rates of the north star-only, proxy-only, and blended decision rules under the practical test.  In the first set of simulations, I vary both the sample size and the structural correlation between the north star and proxy metrics while fixing the true treatment effect on the north star to be $\tau_Y = \sigma_Y$.  I visualize the results using heatmaps where the sample size varies along the horizontal axis and the correlation in treatment effects varies along the vertical axis.

As Figure \ref{fig:sim2-power-heatmaps} shows, at small $n$ and moderate-to-high $\rho$, and when $Y = \tau_Y$, the proxy metric launches with substantially greater probability than does the north star metric: given a true treatment effect on the north star, the proxy metric is much more likely to be significant than the north star metric in this regime.  However, at low $\rho$, the proxy metric can launch with \emph{lower} probability than the north star metric, particularly as the sample size increases, despite its greater signal-to-noise ratio.  The intuition for this is that, when $\rho$ is small, fixing $Y$ at $\tau_Y$ barely pulls the proxy metric away from its prior mean of zero.  As a result, the ``power'' (by the launch definition) of the proxy metric is very low in this regime.  In contrast, the blended metric has good performance in both regimes as it adaptively shifts the balance of decision-making between the north star and proxy metrics.

In the second set of simulations, I fix $\tau_Y = 0$ and plot the false positive risk as $n$ and $\rho$ vary.  As Figure \ref{fig:sim3-fpr-heatmaps} shows, the type I error rate is strictly controlled at 0.05 for the north star metric, whereas it is substantially higher for the proxy metric.  Even more pathologically, the false positive risk is \emph{increasing} in the sample size (except at $\rho = 1$, where the proxy and north star metrics are perfectly aligned), meaning that larger experiments actually incur a worse false positive risk under the proxy-only rule.  In contrast, the blended metric has a much lower false positive risk that interpolates between the north star and proxy-only decision rules.  Furthermore, as the sample size increases, the optimal weights naturally shift away from the proxy metric towards the north star, which gradually reduces the false positive risk.  In the limit, the blended metric converges to the north star metric and so inherits its type I error control.

\begin{figure}
    \centering
    \includegraphics[width=0.9\linewidth]{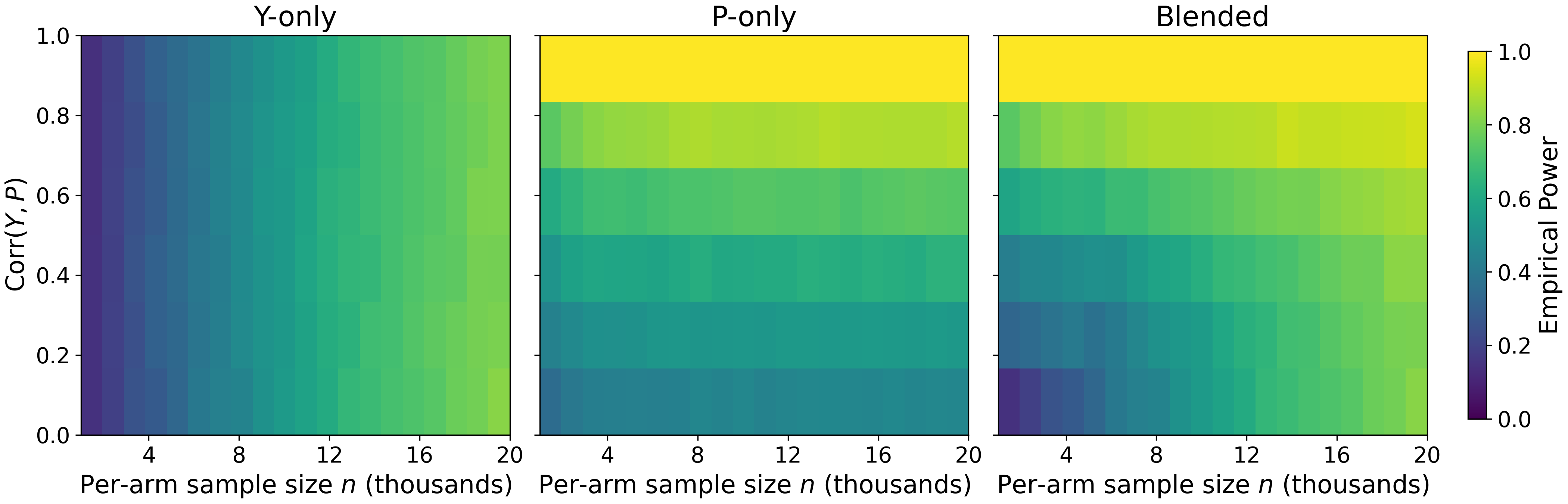}
    \caption{Empirical Power Heatmaps across $n$ and $\rho$.}
    \label{fig:sim2-power-heatmaps}
\end{figure}

\begin{figure}
    \centering
    \includegraphics[width=0.9\linewidth]{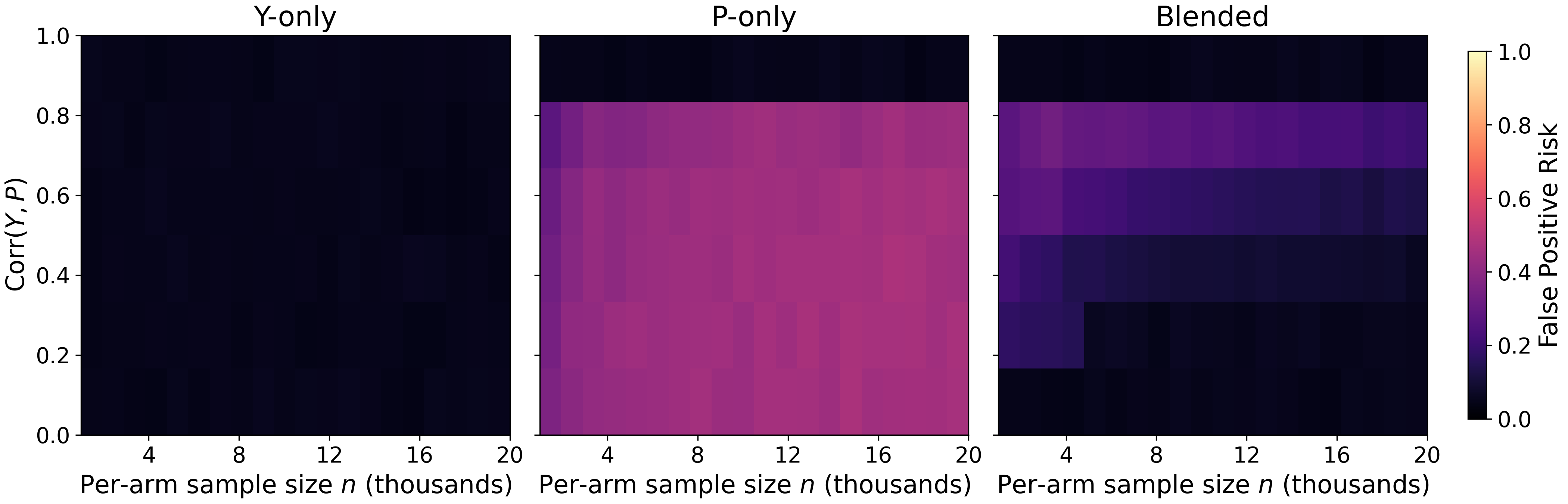}
    \caption{False Positive Risk Heatmaps across $n$ and $\rho$.}
    \label{fig:sim3-fpr-heatmaps}
\end{figure}

Finally, I assess the robustness of the core insights to heavy-tailed treatment effect distributions by repeating the above simulations for a multivariate t-distribution with $\nu=3$ degrees of freedom.  As Figure \ref{fig:sim4-nu3-vs-inf} shows, the qualitative conclusions are robust.  The blended metric consistently achieves better returns than both single-metric rules, with the north star rule converging as the sample size increases.  The power of the proxy metric (again, using the launch definition) remains largely driven by its structural correlation with the north star rather than its sample size.  In contrast, the power of the blended metric increases and even exceeds the power of the proxy as $n$ grows.  Unlike that of the proxy metric, the false positive risk of the blended metric trends downward and converges to the nominal rate as the sample size increases.

\begin{figure}
    \centering
    \includegraphics[width=\linewidth]{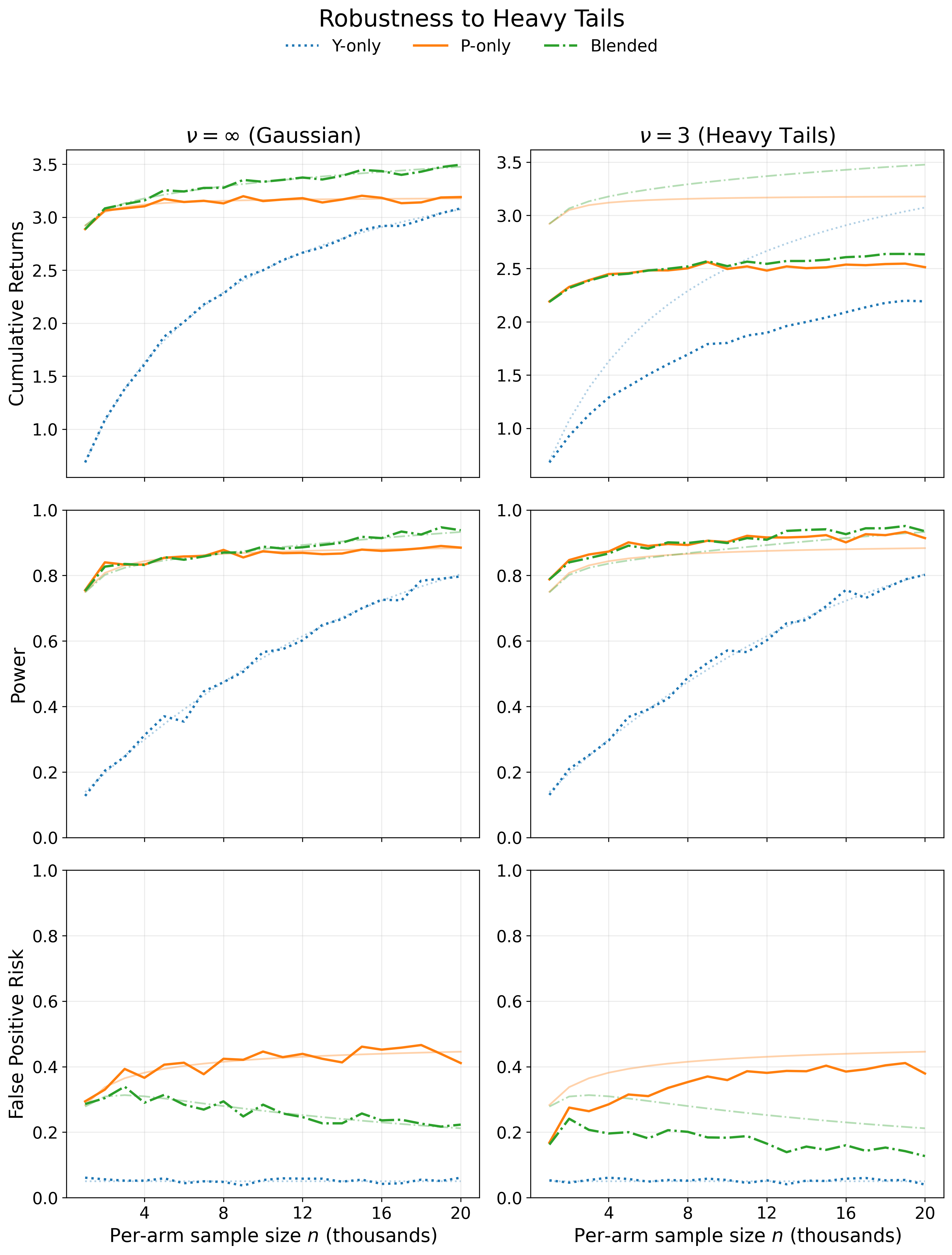}
    \caption{Robustness to Heavy Tails: Returns, Power, and False Positive Risk, Comparing $\nu=\infty$ (Gaussian) vs $\nu=3$.}
    \label{fig:sim4-nu3-vs-inf}
\end{figure}

%% file: sections/05-empirical-application-at-netflix.tex
\section{Empirical Application}
\label{sec:emp}

In this section, I apply the methodology to a real-world experimentation program at Netflix.  The program seeks to optimize, as its north star metric $Y$, a downstream outcome that I will call \emph{plays}.  Plays are preceded by an upstream proxy $P$ that I will refer to as \emph{clicks}.  Clicks are more sensitive to product interventions than are plays, but they are also less tied to long-term engagement and thus less valuable to the program.

The following analysis is based on 15 A/B tests and 50 treatment arms in this program.  All figure axes are anonymized for confidentiality. Although the underlying data are proprietary, code to replicate all calculations and figures can be found at \url{https://github.com/winston-chou/blending}.

\subsection{Assumption Checking}

I begin by verifying that clicks are indeed more sensitive than plays.  Treatment effects on clicks have an SE-ratio---the ratio of the empirical spread of treatment effects to the average within-experiment SE---of $\sim$12 and $\sim$60\% of treatment arms have a statistically significant clicks effect at $\alpha = 0.05$.  In contrast, treatment effects on plays have an SE-ratio of just 2 and only $\sim$25\% of arms are significant.  Thus, in terms of this SE-ratio metric, clicks are about six times more sensitive than plays.

Next, I verify that treatment effects on clicks ($P$) are indeed highly correlated with treatment effects on plays ($Y$).  In Figure~\ref{fig:empirical-covariance}, I plot the observed treatment effects on clicks $\widehat{P}$ against the observed treatment effects on plays $\widehat{Y}$ across all experiments and treatment arms along with their 95\% confidence-interval error bars.  As the plot shows, the two have a strong linear relationship.  To affirm that this relationship is not driven by correlated measurement error, I split each experiment into two folds and estimate the OLS regression of $\widehat{Y}$ on $\widehat{P}$ using the cross-fold estimator \cite{bibaut2024learning, coey2019improving}.  The cross-fold OLS and raw OLS are closely aligned, indicating a limited impact of measurement error.

Lastly, the scatterplot also indicates that the core assumptions of the framework are reasonable: the data are approximately elliptical and the relationship has no obvious nonlinearities or discontinuities.  This suggests that the assumptions of normality and linearity are not too badly violated.

\begin{figure}
    \centering
    \includegraphics[width=0.75\linewidth]{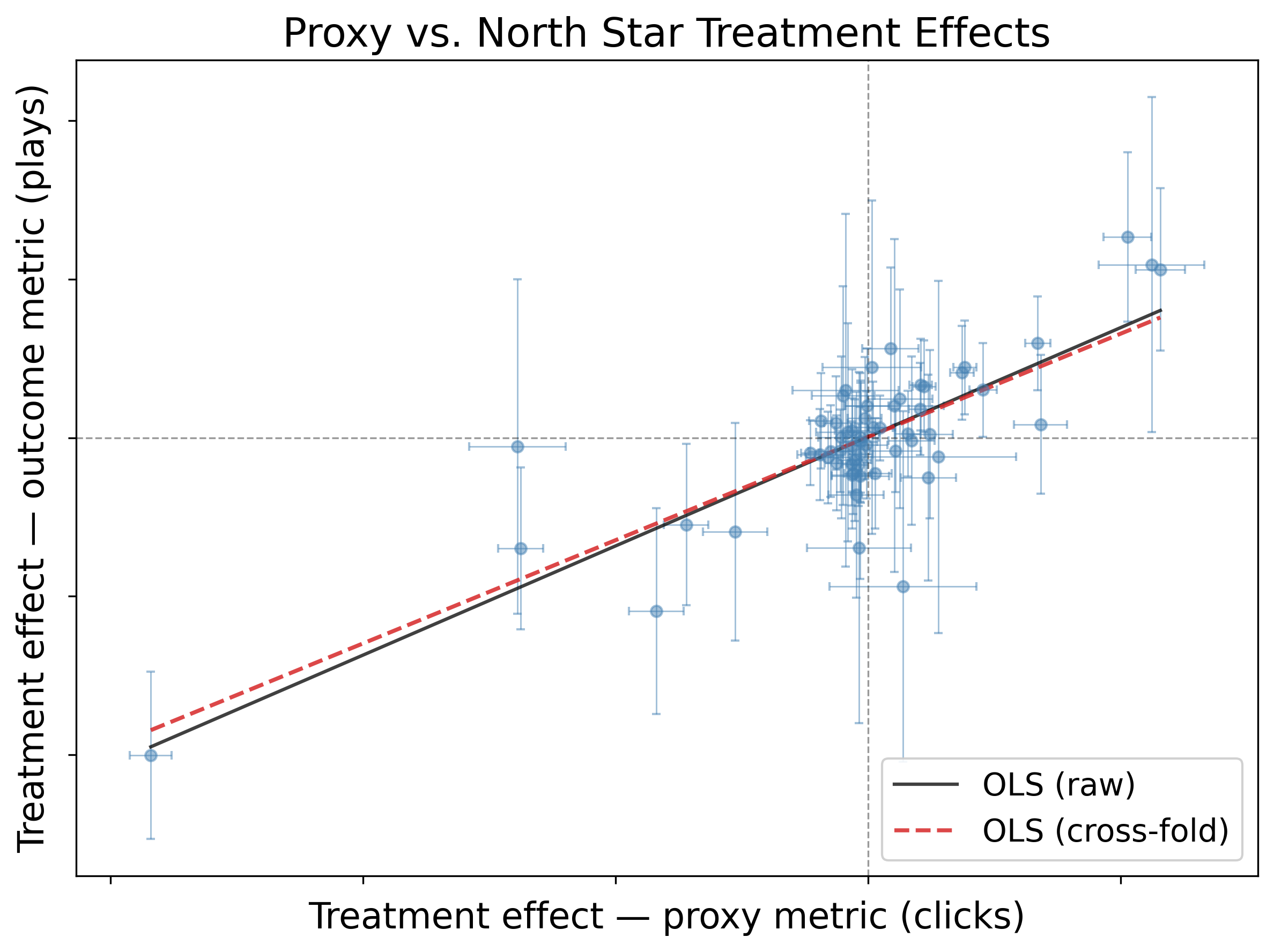}
    \caption{Treatment effects on the proxy (clicks, $x$-axis) vs.\ the north star (plays, $y$-axis) across 15 experiments; error bars are 95\% CIs.  Two fits are shown: raw OLS (black, solid) and cross-fold OLS (red, dashed).  Axes suppressed for confidentiality.}
    \label{fig:empirical-covariance}
\end{figure}

\subsection{Results}

This section presents the optimal blending weights for this dataset.  The minimum sample size test~\eqref{eq:bound-n} is passed for all experiments, and so I use the closed-form approximation to the optimal weights from Proposition~\ref{prop:approximation}.  Figure~\ref{fig:sample-size} displays these weights for the approximate range of sample sizes found in the data ($n \in [1\text{M}, 25\text{M}]$ per arm).  Note that the simplex constraint binds at the lower end of this range, with the weights concentrating entirely on clicks for $n \lesssim 2$ million, in which case the interior approximation is replaced by the simplex-constrained solution (which just puts all weight on clicks).

I find that, at the median experiment size of $n \approx 5$ million per arm, the optimal blend assigns 52\% weight to clicks and 48\% to plays.  The substantial weight on clicks, even at this relatively large sample size, reflects the fact that clicks have a far higher signal-to-noise ratio than do plays.  However, as $n$ grows, the sampling penalty shrinks and $A^*$ tilts toward plays (reaching $\approx 90\%$ plays at the largest experiment sizes $n = 25$ million).

Figure~\ref{fig:returns} displays the expected return per experiment for each decision rule.  The key takeaway is that the blended rule dominates both single-metric rules at all sample sizes observed in the data.  Among the single-metric rules, clicks-only earns higher returns than plays-only across the entire observed range, since the signal-to-noise ratio of plays remains small even at the largest sample sizes; the gap narrows as $n$ grows but does not close below $n = 25$ million.  Yet, as the plot shows, incorporating plays information at these larger sample sizes still improves returns compared to the clicks-only rule.

\begin{figure}
    \centering
    \includegraphics[width=0.75\linewidth]{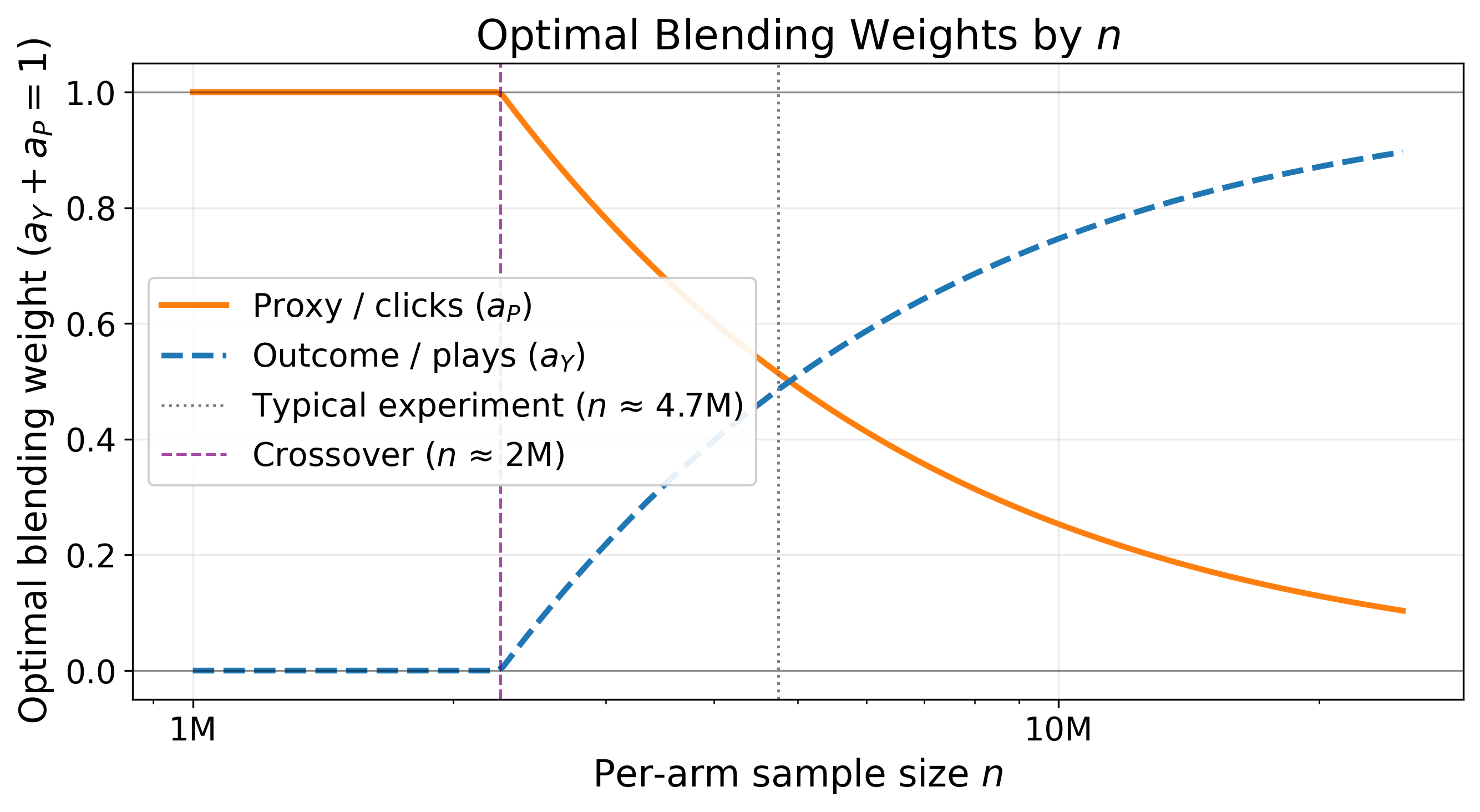}
    \caption{Approximately optimal blending weights vs.\ per-arm sample size $n$ (log scale), using cross-fold parameter estimates from 15 experiments.  The simplex constraint binds below $n \approx 2$ million (weights load entirely on clicks); above it both weights are strictly positive and the weight on plays rises with $n$.  Typical experiment size is approximately 5 million per arm (grey dotted vertical).}
    \label{fig:sample-size}
\end{figure}

\begin{figure}
    \centering
    \includegraphics[width=0.75\linewidth]{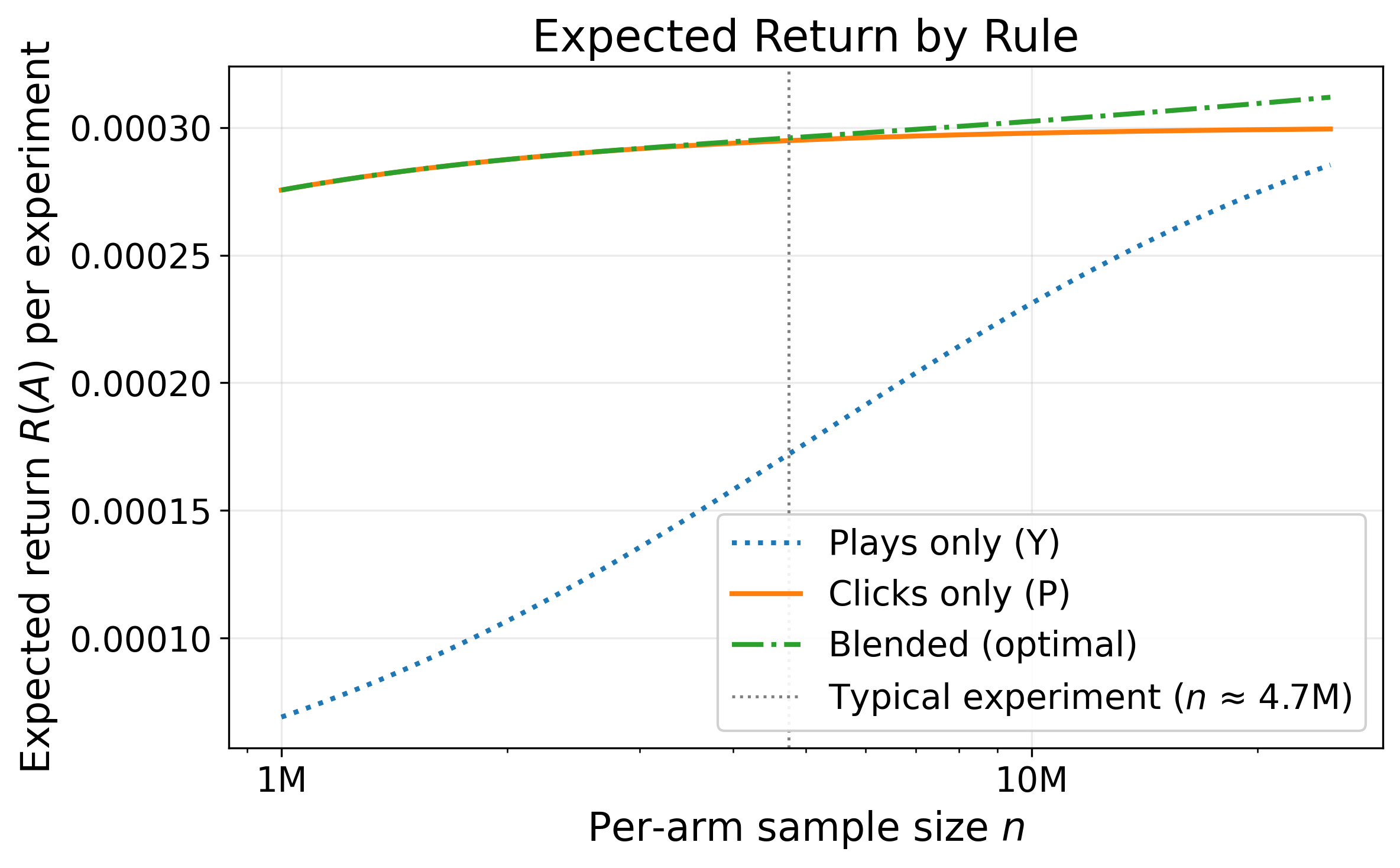}
    \caption{Expected return per experiment vs.\ per-arm sample size $n$ (log scale) for three decision rules: plays-only (blue, dotted), clicks-only (orange, solid), and the optimal blend (green, dash-dot).  The blended rule dominates both single-metric rules and clicks-only earns more than plays-only across the entire range of empirically observed sample sizes.}
    \label{fig:returns}
\end{figure}

%% file: sections/06-conclusion.tex
\section{Conclusion}

I study the optimal role of proxy metrics in decision-making with an observed but statistically insensitive north star.  Rather than basing decisions exclusively on either metric, I propose a blending approach in which the experimenter decides A/B tests using a weighted average of the estimated treatment effects on the proxy metric(s) and the north star.  I provide a closed-form approximation to the optimal blending weights and show that the optimal weight on the north star is increasing in the size of each test and decreasing in the structural alignment of true treatment effects on the north star and the proxy metric(s).  I show how to estimate the requisite parameters for these weights from past experiments, and lastly demonstrate the applicability of the methods to a real-world experimentation program at Netflix.

This work can be extended in several directions.  First, the theoretical framework assumes that experiments are independently and identically distributed.  Relaxing this assumption could lead to novel insights, for example, on how experimenters should allocate resources across different experimentation programs or ideas of varying quality.  Second, I make distributional assumptions that may or may not be appropriate in meta-analyses of past experiments \citep{molenberghs2002statistical}.  In particular, I assume that the underlying treatment effects are Gaussian, when in reality treatment effect distributions may be heavier-tailed \citep{azevedo2019empirical, azevedo2018b}.  Although I find that the core insights are robust to these assumptions, future work should analyze more flexible and realistic models of experiments.  Third, whereas this work assumes access to past experiments, future work can explore how to leverage observational data to inform the design and selection of proxy metrics when no such data are available.

\begin{credits}
    \subsubsection{\ackname} For valuable suggestions, I thank participants in Netflix's Statistics, Methodology, and Engineering (SME) colloquium (2025-08-15) and the four anonymous reviewers.  The Appendix for this paper can be found at \url{https://github.com/winston-chou/blending}.
    
    \subsubsection{\discintname}
    The author has no competing interests to declare that are
    relevant to the content of this article.
\end{credits}

%% file: Supplementary/appendix-content.tex
\appendix

\section{$\underline{n}(\tilde{A})$ is a Conservative Bound}

Let $\underline{n}(A) := (z^2-1)\frac{2\,A^\top\Omega A}{A^\top\Sigma_X A}$.  The goal is to show that $\underline{n}(\tilde{A}) \ge \underline{n}(A^*)$, and so checking that $n \ge \underline{n}(\tilde{A})$ is sufficient to establish that $n \ge \underline{n}(A^*)$, which is a prerequisite to using the closed-form approximation to the optimal weights. 

First, under joint normality of $X = (P, Y)$, the expected returns to a single experiment for any choice of $A$ are given by
\[
R(A) = \frac{\Gamma^\top A}{\sqrt{A^\top \Sigma_\Xhat A}} \phi\left(z \sqrt{q_n(A)}\right),
\]
where $q_n(A) = \frac{A^\top (2\Omega/n) A}{A^\top \Sigma_\Xhat A}$.

Note that the returns are invariant to the scale of $A$ and depend only on its direction.  Therefore, without loss of generality, I can fix the scale of $A$ by requiring $\Gamma^\top A = 1$.  Under this normalization, maximizing $R(A)$ is equivalent to minimizing the following objective function:
\[
F(A) := \log A^\top \Sigma_\Xhat A + z^2 q_n(A).
\]
Abusing notation, I will redefine $A^*$ as the solution to the following constrained minimization problem:
\begin{equation}
    A^* = \arg \min_{\Gamma^\top A = 1} F(A).
\end{equation}

Also under this normalization, $\tilde{A}$ can be written as the solution to the following constrained minimization problem:
\begin{equation}
    \tilde{A} = \arg \min_{\Gamma^\top A = 1} \underbrace{A^\top (\Sigma_X + (1 + z^2) 2 \Omega / n) A}_{W(A)}.
\end{equation}
The following assumes that both $A^*$ and $\tilde{A}$ are interior solutions to their respective minimization problems, i.e., $A_i^* > 0$ and $\tilde{A}_i > 0$ for all $i$.

$W(A)$ can be rewritten as
\begin{eqnarray}
    W(A) &=& A^\top (\Sigma_X + (1 + z^2) 2 \Omega / n) A \\
    &=& A^\top \Sigma_\Xhat A \left(1 + z^2 \frac{A^\top (2\Omega /n) A}{A^\top \Sigma_\Xhat A}\right) \\
    &=& v_A (1 + z^2 q_n(A)),
\end{eqnarray}
where $v_A := A^\top \Sigma_\Xhat A$.

Since $\tilde{A}$ minimizes $W(A)$, we have that:
\begin{eqnarray}
    v_{\tilde{A}}(1 + z^2 q_n(\tilde{A})) &\le& v_{A^*}(1 + z^2 q_n(A^*)) \\ \label{eq:ratio-inequality}
    \frac{v_{\tilde{A}}}{v_{A^*}} &\le& \frac{1 + z^2 q_n(A^*)}{1 + z^2 q_n(\tilde{A})}. 
\end{eqnarray}

Now consider $F(\tilde{A}) - F(A^*)$:
\begin{equation}
    F(\tilde{A}) - F(A^*) = \log \frac{v_{\tilde{A}}}{v_{A^*}} + z^2 (q_n(\tilde{A}) - q_n(A^*)).
\end{equation}

\eqref{eq:ratio-inequality} implies that $F(\tilde{A}) - F(A^*)$ is bounded above by
\begin{equation}
    F(\tilde{A}) - F(A^*) \le \log \frac{1 + z^2 q_n(A^*)}{1 + z^2 q_n(\tilde{A})} + z^2 (q_n(\tilde{A}) - q_n(A^*)).
\end{equation}

Grouping terms, we have that
\begin{equation}
    F(\tilde{A}) - F(A^*) \le [\log (1 + z^2 q_n(A^*)) - z^2 q_n(A^*)] - [\log (1 + z^2 q_n(\tilde{A})) - z^2 q_n(\tilde{A})].
\end{equation}

Suppose, by way of contradiction, that $q_n(A^*) > q_n(\tilde{A})$.  Because $\log(1 + x) - x$ is strictly decreasing on the domain of $q_n(A)$ (i.e., for $x > 0)$, $q_n(A^*) > q_n(\tilde{A})$ implies that the term on the right-hand side of the inequality is negative.  However, $F(\tilde{A}) - F(A^*)$ cannot be negative, since $A^*$ minimizes $F(A)$.  Therefore, it must be the case that $q_n(A^*) \le q_n(\tilde{A})$.  $q_n(A)$ is an increasing function of $\frac{A^\top \Omega A}{A^\top \Sigma_X A}$, so this implies that
\begin{equation}
    \frac{A^{*\top} \Omega A^*}{A^{*\top} \Sigma_X A^*} \le \frac{\tilde{A}^\top \Omega \tilde{A}}{\tilde{A}^\top \Sigma_X \tilde{A}}.
\end{equation}

Multiplying both sides by $2(z^2 - 1)$ yields
\begin{equation}
    \underline{n}(A^*) \le \underline{n}(\tilde{A}),
\end{equation}
as desired.

\section{Average Returns are Decreasing in $n$ for $z \le 2.414$}

The expected returns of a single experiment for any choice of $A$ are given by:
\begin{equation}
    R(A, n) = \frac{\Gamma^\top A}{\sqrt{A^\top \Sigma_\Xhat A}} \phi\left(z \sqrt{q_n(A)}\right),
\end{equation}
where $q_n(A) := \frac{A^\top (2\Omega/n) A}{A^\top \Sigma_\Xhat A}$ is the ratio of the within-experiment sampling variance of $\Zhat(A)$ to its total variance across experiments, $\Sigma_\Xhat = \Sigma_X + 2\Omega/n$.

Note that
\begin{eqnarray}
q_n(A) = q_A(n) &=& \frac{A^\top (2\Omega/n) A}{A^\top \Sigma_\Xhat A} \\
&=& \frac{A^\top (2\Omega) A}{n A^\top \Sigma_X A + A^\top (2\Omega) A},
\end{eqnarray}
so $q_A(n)$ is strictly decreasing in $n$.

Isolating the terms that depend on $n$, we have that:
\begin{equation}
    R(A, n) / n \propto \underbrace{\frac{q_A(n)}{\sqrt{1 - q_A(n)}} \phi\left(z \sqrt{q_A(n)}\right)}_{F_z(q)},
\end{equation}
for $q_A(n) \in (0, 1)$.

Taking logs and differentiating with respect to $q$, we have:
\begin{equation}
    \frac{d}{dq} \log F_z(q) = \frac{1}{q} + \frac{1}{2(1 - q)} - \frac{z^2}{2}.
\end{equation}
The first two terms are minimized at $q^* = 2 - \sqrt{2} \in (0, 1)$ and the minimum value is $1 + \frac{\sqrt{2}}{2} + \frac{1 + \sqrt{2}}{2} = \frac{3}{2} + \sqrt{2} > 0$.  Therefore, the derivative remains positive so long as $z^2 \le 3 + 2\sqrt{2}$.  Set $\sqrt{3 + 2\sqrt{2}} = \sqrt{a} + \sqrt{b}$.

Squaring both sides of the inequality, we have $3 + 2\sqrt{2} = a + b + 2\sqrt{ab}$.  If we set $a + b = 3$ and $2\sqrt{ab} = 2\sqrt{2}$, we have $a = 2$ and $b = 1$.  Therefore, $R / n$ is increasing in $q$ as long as $1 \le z \le 1 + \sqrt{2} \approx 2.414$.  Since $q$ is decreasing in $n$, this implies that $R / n$ is decreasing in $n$ for the same range of $z$ values.

In other words, the average return per unit of allocation is decreasing in $n$ as long as the launch threshold is not too strict.  If it is very strict, then almost no treatment will be launched unless the measurement is sufficiently precise.